\documentclass[12pt,reqno,a4paper]{amsart}
\usepackage{graphicx}

\usepackage{mathtools,amssymb,amsthm}
\usepackage{mathrsfs}
\usepackage{subcaption}
\usepackage[english]{babel}
\usepackage{tikz}
\usetikzlibrary{arrows,shapes,positioning,topaths,intersections,calc,patterns}
\usepackage{booktabs}
\usepackage{etoolbox}
\usepackage{hyperref}

\usepackage[normalem]{ulem}

\usepackage{setspace}
\usepackage{fullpage}
\usepackage{datetime}

\renewcommand{\leq}{\leqslant}
\renewcommand{\geq}{\geqslant}
\newcommand{\eps}{\varepsilon}
\newcommand{\aset}[1]{\{#1\}}

\newcommand{\st}{\;\ifnum\currentgrouptype=16 \middle\fi|\;} 
\DeclareMathOperator*{\argmin}{arg\,min}
\newcommand{\lpare}[1]{\left(#1\right)}

\newcommand{\defi}[1]{{\em #1}}
\newcommand{\cspan}[1]{\ifstrempty{#1}{span}{\operatorname{span}(#1)}}
\newcommand{\strip}{time-strip}

\newcommand{\goodClu}{well-separated}
\newcommand{\bleftmost}{leftmost}
\newcommand{\goodSeq}{$k$-good sequence}
\newcommand{\kgoodSeq}[1]{{#1}-good sequence}

\newcommand{\kcentro}{\textsl{$k$Center}}
\newcommand{\kcsoma}{\textsl{$k$-KinClust1D-SD}} 
\newcommand{\kcmax}{\textsl{$k$-KinClust1D-MD}}

\newcommand{\gp}{\textsc{GreedyPartition}}
\newcommand{\bsearch}{\textsc{$k$-clusteringBS$_{\eps}$}}

\newcommand{\Oh}{\operatorname{O}}
\newcommand{\diam}{\operatorname{diam}}
\newcommand{\sd}{\operatorname{sd}}
\newcommand{\optsd}{\operatorname{sd}^*}
\newcommand{\sdrec}{\operatorname{sd-rec}}
\newcommand{\md}{\operatorname{md}}
\newcommand{\optmd}{\operatorname{md}^*}
\newcommand{\area}{\operatorname{area}}
\newcommand{\Snk}[2]{\genfrac{\lbrace}{\rbrace}{0pt}{}{#1}{#2}}
\newcommand{\dc}{\mathscr{C}}
\newcommand{\restric}[2]{\ensuremath{#1_{#2}}}

\newcommand{\clus}{\mathcal{C}}
\newcommand{\dlus}{\mathcal{D}}

\theoremstyle{definition}
\newtheorem{thm}{Theorem}[section]
\newtheorem{lemma}[thm]{Lemma}

\newtheorem{prob}[thm]{Problem}

\newcommand{\segend}[2]{\fill #1 circle (#2)}
\newcommand{\segendW}[2]{\fill[white] #1 circle (#2); \draw #1 circle (#2);}

\newcommand{\cartplane}[3]{ 
	\draw[gray,->] ($#1+(-0.5, 0)$) -- ($#1+(#2, 0)$);
	\draw[gray,->] ($#1+(0, -0.5)$) -- ($#1+(0, #3+0.5)$);
	\node[right, gray, font=\small] at ($#1+(#2, 0)$) {$x(t)$};
	\node[above, gray, font=\small] at ($#1+(0, #3+0.5)$) {$t$};
	\node[left, gray, font=\small] at (0, #3) {\footnotesize $1$};
}

\newcommand{\Real}{{\mathbb R}}

\title[Kinetic clustering of points on the line]{Kinetic clustering of points on the line}
\author[all]{Cristina G. Fernandes, Marcio T.I. Oshiro}

\dedicatory{\small {\rm Instituto de Matem\'atica e Estat\'{\i}stica \\ Universidade de S\~ao
Paulo, Brazil }}

\thanks{This research has been partially supported by Capes, CNPq
  Projects (Proc. 308523/2012-1), Fapesp Project
  (Proc. 2013/03447-6) and MaCLinC Project of Numec/USP, Brazil.\\
 Email:{\texttt{{\{cris|oshiro\}@ime.usp.br}}}}

\date{\today, \currenttime}

\begin{document}

\begin{abstract}
The problem of clustering a set of points moving on the line consists of the following: given positive integers $n$ and
$k$, the initial position and the velocity of $n$ points, find an optimal $k$-clustering of the points. We consider
two classical quality measures for the clustering: minimizing the sum of the clusters diameters and minimizing the
maximum diameter of a cluster. For the former, we present polynomial-time algorithms under some assumptions and, 
for the latter, a $(2.71 + \eps)$-approximation.
\end{abstract}


\maketitle

\onehalfspacing

\section{Introduction}
\label{sec:intro}

Clustering refers to a well-known class of problems whose goal is to partition a set so that ``similar'' elements are
placed in the same subset of the partition. The notion of similarity and the format of the partition depend on the
application. In this work, we study two clustering problems in a kinetic context, where points move continuously. 

Atallah~\cite{Atallah85} proposed a model for the points movement where the points are in a $d$-dimensional space
and each coordinate of each point is given by a polynomial on the time variable. Using his model, 
Har-Peled~\cite{Har-Peled04} showed how to apply a clustering algorithm for the static setting to find a competitive 
clustering of the moving points. His objective was to find $k$ centers that cover all the points with
minimum radius. When the polynomials describing the points movement have degree at most~$\mu$, his algorithm relaxes 
the restriction on the number of clusters, allowing at most $k^{\mu+1}$ clusters, in order to achieve a constant
approximation ratio with respect to the optimal radius of a $k$-clustering at any time.  

Another model, called KDS (kinetic data structure), is presented by Basch, Guibas, and 
Hershberger~\cite{BaschGuibasHershberger99}. In this model, there is no need to know the full description of the 
points movement, which can be updated online. Using this model, Gao, Guibas, Hershberger, Zhang, and Zhu~\cite{Gao+03} 
proposed a randomized constant approximation to maintain, as the points move, a clustering minimizing the number of 
discrete centers needed to cover all points within a fixed radius.

Lee, Han, and Whang~\cite{LeeHanWhang07} presented a framework for clustering of trajectories, defined as a sequence of
points in a multi-dimensional space describing the movement of an object. They reduced each trajectory to a set of
contiguous line segments and used heuristics to group similar resulting trajectories.

Our work is closer to the one of Har-Peleg, as we look for a static clustering of moving points, instead of a way to 
keep an (almost) optimal $k$-clustering all the time. We consider a restricted case, with points located in~$\Real$, 
linear movements ($\mu=1$), and two classical quality measures for the clustering: minimizing the sum of the clusters 
diameters and minimizing the maximum diameter of a cluster. For the former, we present a polynomial-time algorithm 
under some assumptions and, 
for the latter, a $((4+\sqrt{2})/2 + \eps)$-approximation 
for every~${\eps>0}$.

In Section~\ref{sec:kinetic-model}, we formalize our model for the movement of the points and give the definition of the
diameter of a cluster in our setting, to precisely state the two variants of clustering we address.  In
Section~\ref{sec:min-sum-areas}, we present the polynomial-time algorithm for the first variant and, in
Section~\ref{sec:max-area}, we present the approximation for the second variant and a related open problem.

\section{One dimensional kinetic model and the problems}
\label{sec:kinetic-model}

In our kinetic model, $n$ points move with uniform rectilinear velocity during a continuous time interval. 
Without loss of generality, the time interval is~$[0,1]$. Each point $i \in \aset{1,2,\ldots,n}$ 
has an initial position~$x_i(0)$ and its velocity is given by a vector $v_i$.
We only consider points in $\Real$, so the position and the velocity are real numbers. 
A positive/negative velocity indicates a movement to the right/left respectively. 
This is a particular case of the KDS~\cite{BaschGuibasHershberger99}
and Atallah's model~\cite{Atallah85}.

At an instant $t$ in~$[0,1]$, the position $x_i(t)$ of a 
point~$i$ with initial position~$x_i(0)$ and velocity~$v_i$ is given by
the function \[x_i(t) = x_i(0) + v_it.\] 
This function represents a segment on the Cartesian plane,
called the \defi{trajectory} of point~$i$, and is given by the pair $(x_i(0),v_i)$. 
We draw the Cartesian plane with the horizontal axis
representing the position~$x$ and the vertical axis representing the time $t$.
Since the time interval is always~$[0,1]$, the strip of the plane between~$t = 0$ 
and~$t = 1$ will be called \defi{\strip{}}. 

For our purpose, no two points have the same trajectory, or they can be treated as one. 
Hence, we assume a one-to-one relation between moving points and their trajectories, and
mostly refer to trajectories instead of moving points in what follows. 

Given a finite set $S$ of trajectories, a \defi{cluster} is a subset of $S$ and 
a \defi{$k$-clustering} is a partition of~$S$ into~$k$ clusters. 
Note that, as a cluster might be empty, every $k'$-clustering for $k'< k$ corresponds 
to a $k$-clustering by adding $k-k'$ empty clusters. Conversely, any $k$-clustering 
of $S$ with more than $|S|$ clusters may be converted into an $|S|$-clustering by 
disregarding some empty clusters. So we may assume that $1 \leq k \leq |S|$. 
The \defi{left side} of a nonempty cluster~$C$ is the piecewise linear function
$\min_{i \in C} x_i(t)$ for $t \in [0,1]$. Analogously, the \defi{right side} is 
$\max_{i \in C} x_i(t)$ for $t \in [0,1]$. 
The \defi{\cspan{}} of a cluster~$C$, $\cspan{C}$, is
empty if $C$ is empty, otherwise it is the region within the \strip{} bounded by the left
and right sides of~$C$. The \defi{diameter} of~$C$ is the area of its \cspan{}, denoted by $\diam(C)$. 
See Figure~\ref{fig:line-segments}.
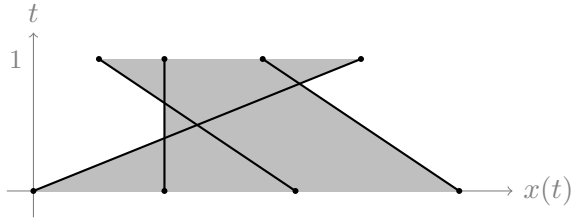
\begin{figure}[ht]
	\center
	\begin{tikzpicture}[scale=0.7]
		\cartplane{(0,0)}{9}{2.5}
		\coordinate (A) at (0, 0);
		\coordinate (B) at (6.15385, 2.5);
		\coordinate (C) at (2.46154, 0);
		\coordinate (D) at (2.46154, 2.5);
		\coordinate (E) at (4.92308, 0);
		\coordinate (F) at (1.23077, 2.5);
		\coordinate (G) at (8, 0);
		\coordinate (H) at (4.30769, 2.5);
		\fill[gray!50] (A) -- (intersection of A--B and C--D) -- 
			(intersection of C--D and E--F) -- (F) -- (B) -- (intersection of A--B and G--H) -- (G) -- cycle;
		\draw[thick] (A) -- (B);
		\draw[thick] (C) -- (D);
		\draw[thick] (E) -- (F);
		\draw[thick] (G) -- (H);
		\segend{(A)}{1.6pt}; \segend{(B)}{1.6pt};
		\segend{(C)}{1.6pt}; \segend{(D)}{1.6pt};
		\segend{(E)}{1.6pt}; \segend{(F)}{1.6pt};
		\segend{(G)}{1.6pt}; \segend{(H)}{1.6pt};
	\end{tikzpicture}
	\caption{\label{fig:line-segments}Representation of a set of moving points by their trajectories.
		The highlighted region is the \cspan{} of the set.}
\end{figure}

The reason to consider the diameter of a cluster as the area of its \cspan{} is because 
we deal with continuous time. 
Usually the diameter of the cluster is the distance between the farthest pair of points in it. 
So, for a continuous time interval, we integrate this distance, as shown 
by Equation~\eqref{eq:integral-area}, which corresponds to the area of the
cluster's span:
\begin{equation}\label{eq:integral-area}
	\int_0^1 \! \left(\max_{i \in C}x_i(t) - \min_{i \in C}x_i(t)\right) \mathrm{d}t \ = \ \area(\cspan{C}) = \diam(C).
\end{equation}

Note that one can calculate the diameter of a cluster $C$ in time polynomial on $|C|$.
If $|C| \leq 1$, then its diameter is 0. Otherwise, we can calculate the diameter of $C$ in 
polynomial time by first getting a
description of its \cspan{} and then using any known algorithm to calculate the area of polygons. 
There are algorithms for this that run in time linear on the number of vertices of the \cspan{}, 
which is linear on the number of trajectories~\cite{ORourke98}.  
The left and right sides of $C$ can be obtained by a divide-and-conquer algorithm that 
runs in~$\Oh(|C| \log |C|)$ time~\cite{Atallah85}.

Two problems are studied in the next sections. The
\defi{sum of diameters kinetic 1D $k$-clustering problem} (\kcsoma)
consists of, given a positive integer~$k$
and a finite set~$S$ of trajectories, finding a $k$-clustering of~$S$ 
whose sum of the clusters diameters is minimized.
The \defi{max diameter kinetic 1D $k$-clustering problem} (\kcmax)
consists of, given a positive integer~$k$
and a finite set~$S$ of trajectories, finding a $k$-clustering of~$S$ 
whose maximum cluster diameter is minimized.
The static 2-dimensional versions of these clustering problems are well-known NP-hard
problems. For both, a~2-approximation is known and is best possible unless P~$=$~NP. 
The static 1D versions of both are polynomially solvable~\cite{Brucker78}.

\section{Minimizing the sum of the diameters}
\label{sec:min-sum-areas}

Let $\clus$ be a $k$-clustering of a set of trajectories. We denote by
\[\sd(\clus) = \sum_{C \in \clus} \diam(C)\]
the sum of the diameters of the clusters in $\clus$. With this, we can formally state our 
first problem as follows.
\begin{prob}[\kcsoma]
	Given a finite set of trajectories $S$ and a positive integer~$k$, find 
	a $k$-clustering $\clus$ of $S$ such that $\sd(\clus)$ is minimum.
\end{prob}

Let $S$ be a finite set of $n$ trajectories. Let $\clus^*$ be an optimal $k$-clustering of $S$,
i.e., $\sd(\clus^*)$ is minimum among all $k$-clusterings of $S$. We denote by 
$\optsd(S,k) = \sd(\clus^*)$ the value of an optimal $k$-clustering, which is, in this case,
the sum of the diameters of the clusters in $\clus^*$.

For the next lemma, recall that we assumed that $k \leq n$. 

\begin{lemma}\label{lemma:opt-nonempty}
	Let $S$ be a set of $n$ trajectories. There is always an optimal
	$k$-clustering of $S$ for the \kcsoma{} without empty clusters.
\end{lemma}
\begin{proof}
	Let $\clus^*$ be an optimal $k$-clustering of $S$ with minimal number
	of empty clusters. Suppose that~$\clus^*$ has at least one empty cluster.
	Let~$C$ be a non-empty	cluster of~$\clus^*$ with
	at least two trajectories. Such cluster must exist as $k \leq n$.

	Take some trajectory~$s$ from the left side of~$C$. Remove~$s$ from~$C$
	and replace an empty cluster by a cluster containing only~$s$. This
	does not increase the diameter of~$C$ and the new cluster has diameter
	zero. Thus, we obtained a $k$-clustering~$\clus'$ with $\sd(\clus') \leq \sd(\clus^*)$, 
        which means  $\sd(\clus') = \sd(\clus^*)$ because $\clus^*$ is optimal, 
	and with fewer empty clusters, contradicting the choice of~$\clus^*$.
	Hence, there is always an optimal $k$-clustering of $S$ for the 
	\kcsoma{} without empty clusters.
\end{proof}

By Lemma~\ref{lemma:opt-nonempty}, we may always look for an optimal
$k$-clustering without empty clusters. However, $k$-clusterings with
empty clusters can be used to bound the value of~$\optsd(S,k)$.

For any integer $k$ with $1 \leq k \leq n$, the number of
distinct $k$-clusterings of~$S$ without empty clusters is given by the Stirling number of the second kind,
denoted by~$\Snk{n}{k}$.  It is known~\cite{RennieDobson69} that  
 \[ \Snk{n}{k} \ = \ \frac{1}{k!} \sum_{i=0}^{k}(-1)^{k-i}\binom{k}{i}i^n \ \geq \ \frac{1}{2}(k^2 + k + 2)k^{n-k-1}-1 \ = \ \Omega(k^{n-k+1}).\] 
For $k \geq 2$, this number is exponential on the number of trajectories even if $k$ is a fixed value.
So, it would take too long to examine every $k$-clustering to find an optimal one, for~$k \geq 2$.

The trajectories in $S$ divide the \strip{} into convex polygonal regions.
We call each of these regions a \defi{hole} (of $S$).
Note that a hole always has a positive area, thus its interior is never empty. 
The set of holes of $S$ is denoted by $H(S)$, or simply $H$ if $S$ is
clear from the context.

Let $s$ be a trajectory in~$S$ and~$h$ be a hole. 
We say that~$s$ is \defi{to the left} of~$h$ if there exists a point~$(x',t')$ inside~$h$
such that~$x_s(t') < x'$. Otherwise,~$s$ is \defi{to the right} of~$h$. Intuitively, we
trace a horizontal line passing through~$h$. Then,~$s$ is to the left of~$h$ if, in this
line, the point intersecting~$s$ is to the left of some point of~$h$. 
See Figure~\ref{fig:hole-separating}. 
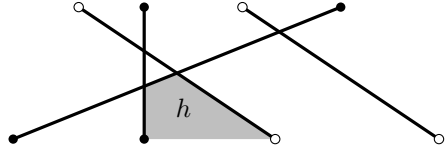
\begin{figure}[ht]
	\center
	\begin{tikzpicture}[scale=0.7]
		\coordinate (A) at (0, 0);
		\coordinate (B) at (6.15385, 2.5);
		\coordinate (C) at (2.46154, 0);
		\coordinate (D) at (2.46154, 2.5);
		\coordinate (E) at (4.92308, 0);
		\coordinate (F) at (1.23077, 2.5);
		\coordinate (G) at (8, 0);
		\coordinate (H) at (4.30769, 2.5);
		\fill[gray!50] (C) -- (intersection of A--B and C--D) -- 
			(intersection of A--B and E--F) -- (E) -- cycle;
		\draw[very thick] (A) -- (B); \segend{(A)}{2.5pt}; \segend{(B)}{2.5pt};
		\draw[very thick] (C) -- (D); \segend{(C)}{2.5pt}; \segend{(D)}{2.5pt};
		\draw[very thick] (E) -- (F); \segendW{(E)}{2.5pt}; \segendW{(F)}{2.5pt};
		\draw[very thick] (G) -- (H); \segendW{(G)}{2.5pt}; \segendW{(H)}{2.5pt};
		\node at (3.2,0.6) {\footnotesize $h$};
		\end{tikzpicture} 
	\caption{\label{fig:hole-separating}Trajectories with black ends
	are to the left of $h$ and trajectories with white ends
	are to the right.}
\end{figure}

Note that each trajectory is either to the left or to the right of $h$.
Thus, each hole $h$ partitions $S$ into two parts $S_r(h)$ and
$S_{\ell}(h)$, with the trajectories in $S$ to the right and to the left of~$h$ respectively. 
Given a $k$-clustering $\clus$ of~$S$, a hole that is contained in the \cspan{} of some cluster in $\clus$ 
is \defi{covered} by $\clus$. Otherwise it is \defi{uncovered} by $\clus$. 
A hole~$h$ \defi{separates} distinct clusters~$C_1$ and~$C_2$ in~$\clus$
if $C_1 \subseteq S_{\ell}(h)$ and $C_2 \subseteq S_r(h)$, 
or $C_1 \subseteq S_{r}(h)$ and~$C_2 \subseteq S_{\ell}(h)$. The reference to the $k$-clustering
is omitted when it is clear.

Since we only consider finite sets of trajectories, the number of holes in
these sets is not only finite, but polynomial on the number of trajectories.
This is a straightforward consequence, 
when we extend the trajectories to lines, of the fact that a set of~$n$
lines divides the plane in at most~${n(n+1)/2 + 1}$ regions. Also, for any
finite set of trajectories, there are always two \defi{unbounded} holes, 
$h_{-}$ and $h_{+}$, whose region extends infinitely to the left and to the
right, respectively. The other holes are called \defi{bounded}.

\begin{lemma}\label{lemma:sum-separation}
	Let $S$ be a set of trajectories and let $\clus^*$ be an optimal $k$-clustering of $S$.
	For any two distinct clusters of $\clus^*$, there is a hole of~$S$ separating them.
\end{lemma}
\begin{proof}
	Suppose there are clusters $C_1$ and $C_2$ in $\clus^*$ such that there is no hole of~$S$
	separating them. This means that $\cspan{C_1} \cap \cspan{C_2} \neq \emptyset$
	and, consequently, ${\cspan{C_1} \cup \cspan{C_2} = \cspan{C_1 \cup C_2}}$,
	otherwise there would be a hole of~$S$ separating~$C_1$ and~$C_2$. 
	Hence,
	$\diam(C_1) + \diam(C_2) > \diam(C_1 \cup C_2)$. So, merging~$C_1$ and~$C_2$ 
	would result in a better $k$-clustering. 
	Therefore, any two distinct clusters in~$\clus^*$ must be separated by a hole of~$S$.
\end{proof}

Notice that Lemma~\ref{lemma:sum-separation} does not guarantee that there is always an optimal
$k$-clustering with $k-1$ uncovered holes of $S$, since all holes separating two clusters could be
covered by other clusters, as shown in Figure~\ref{fig:bad-separated}.
However, we can guarantee the existence of at least one 
uncovered hole in an optimal $k$-clustering, for $k > 1$.

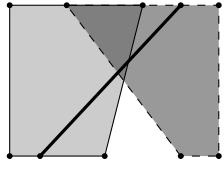
\begin{figure}[ht]
	\begin{center}
		\begin{tikzpicture}[scale=0.5]
			\coordinate (A0) at (0.5, 0);
			\coordinate (A1) at (0.5, 4);
			\coordinate (B0) at (1.3, 0);
			\coordinate (B1) at (5, 4);
			\coordinate (C0) at (3, 0);
			\coordinate (C1) at (4, 4);
			\coordinate (D0) at (5, 0);
			\coordinate (D1) at (2, 4);
			\coordinate (E0) at (6, 0);
			\coordinate (E1) at (6, 4);

			\fill[gray!80, opacity=0.5] (A0) -- (C0) -- (C1) -- (A1) -- cycle;
			\fill[black!80, opacity=0.5] (D0) -- (E0) -- (E1) -- (D1) -- cycle;
			\segend{(A0)}{2.2pt}; \segend{(A1)}{2.2pt};
			\segend{(B0)}{2.2pt}; \segend{(B1)}{2.2pt};
			\segend{(C0)}{2.2pt}; \segend{(C1)}{2.2pt};
			\segend{(D0)}{2.2pt}; \segend{(D1)}{2.2pt};
			\segend{(E0)}{2.2pt}; \segend{(E1)}{2.2pt};
			\draw[thin] (A0) -- (A1) -- (C1) -- (C0) -- cycle;
			\draw[very thick] (B0) -- (B1);
			\draw[dashed, thin] (D0) -- (D1) -- (E1) -- (E0) -- cycle;
		\end{tikzpicture}
	\end{center}
	\caption{\label{fig:bad-separated}A not \goodClu{} $3$-clustering.}
\end{figure}

Let $S$ be a finite set of~$n$ trajectories. We say that a trajectory~$s$ in~$S$ is
\defi{\bleftmost{}} if~$x_s(0)$ is minimum in~$S$ and, in case of ties, 
$x_s(1)$ is minimum within the tied trajectories.

\begin{lemma}\label{lemma:sum-separation-uncovered}
	Let $S$ be a set of trajectories and let $\clus^*$ be an optimal $k$-clustering of $S$.
	For $k > 1$, there is at least one uncovered hole of~$S$.
\end{lemma}
\begin{proof}
	Suppose that every hole of $S$ is covered by $\clus^*$. Then,
	\[\sd(\clus^*) \ = \ \sum_{C \in \clus^*} \diam(C) \ \geq \ \sum_{h \in H} \area(h).\]
	Let $s$ be the \bleftmost{} trajectory of $S$. There is a hole $h_s$ in~$H$ separating $\aset{s}$
	and~$S \setminus \aset{s}$. Consider the $k$-clustering~$\clus$ with the clusters $\aset{s}$, 
	$S \setminus \aset{s}$, and the remaining clusters empty. Then,
	\[\sd(\clus) \ = \ \sum_{h \in H} \area(h) - \area(h_s) \ < \ \sum_{h \in H} \area(h) \ \leq \ \sd(\clus^*).\]
	This is a contradiction, since $\clus^*$ is an optimal $k$-clustering. Therefore, there is at least one
	hole uncovered by $\clus^*$.
\end{proof}

Lemmas~\ref{lemma:sum-separation} and~\ref{lemma:sum-separation-uncovered} state some properties
of an optimal $k$-clustering. So, we do not need to consider every possible $k$-clustering
in order to find an optimal one. Considering
only the $k$-clusterings with such properties is enough. 
To do that, we define a \defi{\goodSeq{}} for~$S$
as a sequence $\aset{(h_i,C_i)}_{i=1}^{k-1}$ such that, for $1 \leq i \leq k-1$ 
and~$\clus_0 = \aset{S}$, we have 
\begin{align}
	C_i & \in \clus_{i-1} \nonumber \\
	h_i & \in H \mbox{ contained in }\cspan{C_i} \label{eq:sequence-hole-cluster}\\
	\clus_i & = (\clus_{i-1} \setminus \aset{C_i}) \cup \aset{(C_i)_{\ell}(h_i),(C_i)_{r}(h_i)}.\nonumber
\end{align}
For $i \geq 1$, each $\clus_i$ is an $(i+1)$-clustering obtained by separating the
cluster~$C_i$ of~$\clus_{i-1}$ using the hole~$h_i$. 
Hence, a \goodSeq{} defines a $k$-clustering.

\begin{thm}\label{thm:sum-optimal-sequence}
	Let $S$ be a set of trajectories and~$k \geq 1$. 
	Every optimal $k$-clustering of~$S$
	is defined by a \goodSeq{}.
\end{thm} 
\begin{proof}
	The proof is by induction on $k$. For $k = 1$, the statement is trivial, since
	$\clus_0 = \aset{S}$ is the only $1$-clustering of~$S$. So, consider $k > 1$
	and let $\clus^*$ be an optimal $k$-clustering of~$S$.

	By Lemma~\ref{lemma:sum-separation-uncovered}, there is an uncovered
	hole~$h_1$ in~$H$.
	Then,~$h_1$ separates~$S$ into~$S_1 = S_{\ell}(h_1)$ and~$\bar{S}_1 = S_r(h_1)$,
	so that $C \subseteq S_1$ or~$C \subseteq \bar{S}_1$, for every $C \in \clus^*$. 
	Let $\dlus = \aset{C \in \clus^* \st C \subseteq S_1}$, $k_1 = |\dlus|$, 
	$\bar{\dlus} = \aset{C \in \clus^* \st C \subseteq \bar{S}_1}$, and $\bar{k}_1 = |\bar{\dlus}|$.
	Note that $k_1 \geq 1$, $\bar{k}_1 \geq 1$, and $k_1 + \bar{k}_1 = k$.
	Moreover,~$\dlus$ is an optimal $k_1$-clustering of $S_1$ and $\bar{\dlus}$ is an
	optimal $\bar{k}_1$-clustering of $\bar{S}_1$, otherwise the $k$-clustering 
        $\clus^* = \dlus \cup \bar{\dlus}$ would not be optimal for $S$.

	By induction hypothesis, there is a \kgoodSeq{$k_1$}~$\aset{(d_i,C_i)}_{i=1}^{k_1-1}$
	defining~$\dlus$ and there is a \kgoodSeq{$\bar{k}_1$} $\aset{(\bar{d}_i,\bar{C}_i)}_{i=1}^{\bar{k}_1-1}$
	defining $\bar{\dlus}$. Note that $d_i$ is a hole of $S_1$. If~$d_i$ is not also a hole of~$S$,
	then~$d_i$ corresponds to the union of two or more holes of $S$. Any one of these holes of
	$S$ inside $d_i$ separates $S_1$ in the same way as $d_i$. Thus, we can exchange $d_i$ by one of the
	holes of $S$ inside it. The same can be done to each hole $\bar{d}_i$ of $\bar{S}_1$.

	Therefore, $(h_1,S),(d_1,C_1),\ldots,(d_{k_1-1},C_{k_1-1}),(\bar{d}_1,\bar{C}_1),
	\ldots,(\bar{d}_{\bar{k}_1-1},\bar{C}_{\bar{k}_1-1})$ is a \goodSeq{} that defines $\clus^*$.
\end{proof}

\begin{lemma}\label{lemma:number-sequences-hole-cluster}
	Let $S$ be a set of $n$ trajectories and~$k \geq 1$.
	The number of \goodSeq{}s of~$S$
	is $\Oh(n^{2(k-1)}(k-1)!)$.
\end{lemma}
\begin{proof}
	Remember that $|H| = \Oh(n^2)$.
	For each element $(h_i, C_i)$ of a \goodSeq{}, there are at most $|H| - i + 1$
	choices for hole $h_i$, and at most~${|\clus_{i-1}| = i}$ choices for cluster $C_i$.
	Thus, the number of such sequences is at most
	\[\prod_{i = 1}^{k-1} \left(|H| - i + 1\right)i 
			\ = \ \Oh\!\!\left(\prod_{i = 1}^{k-1} n^2i \right) 
			\ = \ \Oh\!\left(n^{2(k-1)}(k-1)!\right).\qedhere\]
\end{proof}

So, instead of considering all the $\Snk{n}{k} = \Omega(k^{n-k})$ possible
$k$-clusterings of $S$, we may consider only those that are defined by
a \goodSeq{}. Lemma~\ref{lemma:number-sequences-hole-cluster} guarantees
that the number of such sequences if polynomial in the input size for a
fixed value of~$k$. Moreover, since the diameter of a cluster with
$m$ trajectories can be calculated in time $\Oh(m \log m)$, an
algorithm that examines every $k$-clustering defined by a \goodSeq{}
runs in time $\Oh((k-1)!\ n^{2k-1} \log n)$.

\subsection{Well-separated clusterings}\label{sec:well-separated-sum}
Depending on the application, we may want the similarity between
distinct clusters to be as small as possible. For example, a cluster whose \cspan{} 
is contained in the \cspan{} of another one may be undesirable even if this configuration
was the only one to minimize the objective function.
We say that a $k$-clustering is \defi{\goodClu{}} if its clusters are pairwise separated by an uncovered hole.
Figure~\ref{fig:bad-separated} shows an example of a $3$-clustering that is not \goodClu{}, 
because there is no uncovered hole separating the bold
trajectory from the light gray cluster.

If we require the holes $b_i$ of a \goodSeq{} to be all uncovered, then the algorithm 
described in Lemma~\ref{lemma:number-sequences-hole-cluster} would find an optimal 
\goodClu{} $k$-clustering for \kcsoma{}.  However, we can do better using a dynamic 
programming algorithm that is polynomial in the input size, even if $k$ is part of the input. 

Let ${\dc_H = \cup_{h \in H} \aset{S_{\ell}(h),S_r(h)}}$.
Consider the partial order $\preceq$ over $\dc_H$ defined as follows:
for every $C_1$ and~$C_2$ in~$\dc_H$, we say $C_1 \preceq C_2$ 
if and only if $C_1 \subseteq C_2$.
We denote by $D_H$ the directed acyclic graph (dag)
representing this partial order. See Figure~\ref{fig:partial-order-dag}.
Notice that $|V(D_H)| = |\dc_H| \leq 2|H|$.
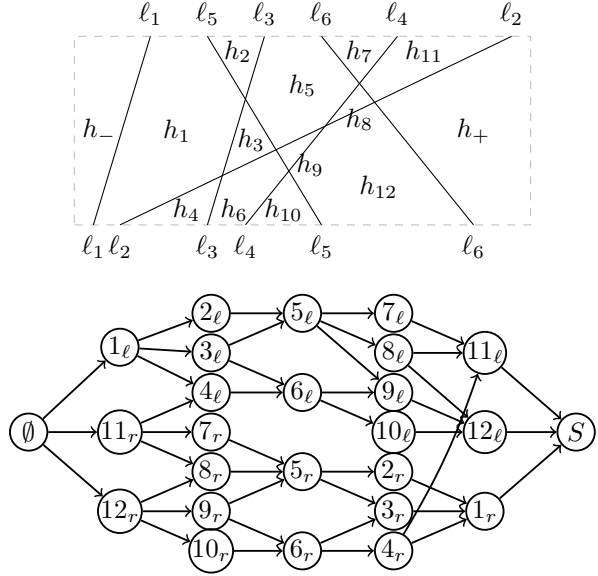
\begin{figure}[ht]
	\begin{subfigure}[]{\textwidth}
		\center
		\begin{tikzpicture}[scale=0.5]
			\coordinate (A1) at (0, 0);
			\coordinate (A2) at (1.5, 5);
			\coordinate (B1) at (0.7, 0);
			\coordinate (B2) at (11, 5);
			\coordinate (C1) at (3,0);
			\coordinate (C2) at (4.5, 5);
			\coordinate (D1) at (4, 0);
			\coordinate (D2) at (8, 5);
			\coordinate (E1) at (6, 0);
			\coordinate (E2) at (3, 5);
			\coordinate (F1) at (10, 0);
			\coordinate (F2) at (6, 5);

			\node[below] at (A1) {\footnotesize $\ell_1$};
			\node[below] at (B1) {\footnotesize $\ell_2$};
			\node[below] at (C1) {\footnotesize $\ell_3$}; 
			\node[below] at (D1) {\footnotesize $\ell_4$};
			\node[below] at (E1) {\footnotesize $\ell_5$};
			\node[below] at (F1) {\footnotesize $\ell_6$};
			\node[above] at (A2) {\footnotesize $\ell_1$};
			\node[above] at (B2) {\footnotesize $\ell_2$};
			\node[above] at (C2) {\footnotesize $\ell_3$};
			\node[above] at (D2) {\footnotesize $\ell_4$};
			\node[above] at (E2) {\footnotesize $\ell_5$};
			\node[above] at (F2) {\footnotesize $\ell_6$};
			
			\draw[gray!50, dashed] (-0.5,0) rectangle (11.5,5);
			\draw (A1) -- (A2);
			\draw (B1) -- (B2);
			\draw (C1) -- (C2);
			\draw (D1) -- (D2);
			\draw (E1) -- (E2);
			\draw (F1) -- (F2);

			\node[] at (0.15, 2.5) {\footnotesize$h_{-}$}; 
			\node[] at (2.2, 2.5) {\footnotesize$h_1$}; 
			\node[] at (3.8, 4.6) {\footnotesize$h_2$}; 
			\node[] at (4.15, 2.25) {\footnotesize$h_3$}; 
			\node[] at (2.45, 0.4) {\footnotesize$h_4$}; 
			\node[] at (5.5, 3.7) {\footnotesize$h_5$}; 
			\node[] at (3.7, 0.4) {\footnotesize$h_6$}; 
			\node[] at (7, 4.6) {\footnotesize$h_7$}; 
			\node[] at (7, 2.8) {\footnotesize$h_8$}; 
			\node[] at (5.7, 1.6) {\footnotesize$h_9$}; 
			\node[] at (5, 0.4) {\footnotesize$h_{10}$}; 
			\node[] at (8.7, 4.6) {\footnotesize$h_{11}$}; 
			\node[] at (7.5, 1) {\footnotesize$h_{12}$}; 
			\node[] at (10, 2.5) {\footnotesize$h_{+}$};
		\end{tikzpicture}
	\end{subfigure}
	
	\vspace*{0.3cm}
	\begin{subfigure}[]{\textwidth}
		\center
		\begin{tikzpicture}[scale=0.3, minimum size=14, inner sep=0pt, line width=0.7pt]
			\node[draw, circle] (h13) at (24,0) {\footnotesize$S$};
			
			\node[draw, circle] (l11) at (20,3.5) {\footnotesize$11_{\ell}$}
				edge[->] node {} (h13);
			\node[draw, circle] (l12) at (20,0) {\footnotesize$12_{\ell}$}
				edge[->] node {} (h13);
			\node[draw, circle] (r1) at (20,-3.5) {\footnotesize$1_{r}$}
				edge[->] node {} (h13);

			\node[draw, circle] (l7) at (16,5.25) {\footnotesize$7_{\ell}$}
				edge[->] node {} (l11);
			\node[draw, circle] (l8) at (16,3.5) {\footnotesize$8_{\ell}$}
				edge[->] node {} (l12)
				edge[->] node {} (l11);
			\node[draw, circle] (l9) at (16,1.75) {\footnotesize$9_{\ell}$}
				edge[->] node {} (l12); 
			\node[draw, circle] (l10) at (16,0) {\footnotesize$10_{\ell}$}
				edge[->] node {} (l12);
			\node[draw, circle] (r2) at (16,-1.75) {\footnotesize$2_{r}$}
				edge[->] node {} (r1);
			\node[draw, circle] (r3) at (16,-3.5) {\footnotesize$3_{r}$}
				edge[->] node {} (r1); 
			\node[draw, circle] (r4) at (16,-5.25) {\footnotesize$4_{r}$}
				edge[->] node {} (r1)
				edge[->, bend right=6] node {} (l11);

			\node[draw, circle] (l5) at (12,5.25) {\footnotesize$5_{\ell}$}
				edge[->] node {} (l9)
				edge[->] node {} (l8)
				edge[->] node {} (l7); 
			\node[draw, circle] (l6) at (12,1.75) {\footnotesize$6_{\ell}$}
				edge[->] node {} (l10)
				edge[->] node {} (l9); 
			\node[draw, circle] (r5) at (12,-1.75) {\footnotesize$5_{r}$}
				edge[->] node {} (r2)
				edge[->] node {} (r3); 
			\node[draw, circle] (r6) at (12,-5.25) {\footnotesize$6_{r}$}
				edge[->] node {} (r3)
				edge[->] node {} (r4);

			\node[draw, circle] (l2) at (8,5.25) {\footnotesize$2_{\ell}$}
				edge[->] node {} (l5);
			\node[draw, circle] (l3) at (8,3.5) {\footnotesize$3_{\ell}$}
				edge[->] node {} (l5)
				edge[->] node {} (l6);
			\node[draw, circle] (l4) at (8,1.75) {\footnotesize$4_{\ell}$}
				edge[->] node {} (l6); 
			\node[draw, circle] (r7) at (8,0) {\footnotesize$7_{r}$}
				edge[->] node {} (r5);
			\node[draw, circle] (r8) at (8,-1.75) {\footnotesize$8_{r}$}
				edge[->] node {} (r5);
			\node[draw, circle] (r9) at (8,-3.5) {\footnotesize$9_{r}$}
				edge[->] node {} (r5)
				edge[->] node {} (r6); 
			\node[draw, circle] (r10) at (8,-5.25) {\footnotesize$10_{r}$}
				edge[->] node {} (r6);

			\node[draw, circle] (l1) at (4,3.75) {\footnotesize$1_{\ell}$}
				edge[->] node {} (l2)
				edge[->] node {} (l3)
				edge[->] node {} (l4);
			\node[draw, circle] (r11) at (4,0) {\footnotesize$11_{r}$}
				edge[->] node {} (l4)
				edge[->] node {} (r7)
				edge[->] node {} (r8); 
			\node[draw, circle] (r12) at (4,-3.5) {\footnotesize$12_{r}$}
				edge[->] node {} (r8)
				edge[->] node {} (r9)
				edge[->] node {} (r10);

			\node[draw, circle] (h0) at (0,0) {\footnotesize$\emptyset$}
				edge[->] node[left, black] {} (l1)
				edge[->] node[right, black] {} (r11)
				edge[->] node[left, black] {} (r12);
		\end{tikzpicture}
	\end{subfigure}
	\caption{\label{fig:partial-order-dag}Dag $D_H$. 
	Vertices $i_{\ell}$ and $i_r$ represent
	$S_{\ell}(h_i)$ and $S_r(h_i)$, respectively.}
\end{figure}

\begin{lemma}\label{lemma:separating-holes}
	Let $S$ be a set of trajectories and $\clus$ be a \goodClu{} $k$-clustering of $S$.
	There is an ordering $C_1,C_2,\ldots,C_k$ of the clusters in $\clus$ such that,
	for every $1 \leq i \leq k-1$, 
	\[\bigcup_{j=1}^i C_j \quad \mbox{ and } \quad \bigcup_{j=i+1}^k C_j\]
	are separated by some uncovered hole $h_i$ of $H$.
\end{lemma}
\begin{proof}
	Let $\tilde{H}$ be the set of holes of~$S$ uncovered by~$\clus$.
	Suppose that~$k \geq 2$, otherwise the statement is trivial.
	Consider the following construction. Let 
	\[	A_1 = \argmin_{A \in \dc_{\tilde{H}}} |A|, \quad\textrm{and}\quad
		A_i = \argmin_{\substack{A \in \dc_{\tilde{H}},\\A_{i-1} \subset A}} |A|,
		\textrm{ for } 2 \leq i \leq k-1. \]%
	Also, consider a hole $h_i \in \tilde{H}$ such 
	that $A_i \in \aset{S_{\ell}(h_i), S_r(h_i)}$, for ${1 \leq i \leq k-1}$.

	Take $C_1 = A_1$, $C_i = A_i \setminus A_{i-1}$, for each $2 \leq i \leq k-1$,
	and~$C_k = S \setminus A_{k-1}$.
	Observe that~$\bigcup_{j=1}^i C_j$ and $\bigcup_{j=i+1}^k C_j$ are separated by~$h_i$,
	for all $1 \leq i \leq k-1$. Now, we only need to show that this construction is well-defined,
	i.e., none of the sets is empty and~$\clus = \aset{C_1, C_2, \ldots, C_k}$.

	Since~$\clus$ is \goodClu{}, $|\tilde{H}| \geq k - 1 \geq 1$.
	Hence,~$C_1$ exists and is nonempty. Moreover,~$C_1$ is a cluster of~$\clus$, otherwise
	there would be at least two clusters of~$\clus$ in~$C_1$. However, these two clusters would
	have to be separated by an uncovered hole~$h$, implying that $\min(|S_{\ell}(h)|,|S_r(h)|) < |C_1|$.
	This contradicts the choice of~$C_1$.

	Fix $2 \leq i \leq k-1$.
	Suppose that $C_1,C_2,\ldots,C_{i-1}$ are nonempty and that they are clusters of~$\clus$.
	Since~$\clus$ is a \goodClu{} $k$-clustering, the other~$k-i+1 \geq 2$ clusters of~$\clus$
	are inside~$S \setminus A_{i-1}$. Hence, there is a hole~$h$ in~$\tilde{H}$ separating some
	pair of clusters inside~$S \setminus A_{i-1}$. Thus, we can conclude that~$A_{i-1}$ is
	properly contained in~$S_{\ell}(h)$ or in~$S_r(h)$, implying that~$A_i$ and~$C_i$ are nonempty.
	Since~$A_i$ has minimum cardinality and~$A_{i-1} \subset A_i$, by the same argument used previously
	for~$C_1$, the set~$A_{i} \setminus A_{i-1}$ indeed is a cluster of~$\clus$.
	Therefore, ${\clus = \aset{C_1,C_2,\ldots,C_k}}$.
\end{proof}

A \defi{$k$-chain} is a sequence $\aset{C_i}_{i=1}^{k-1}$ of $k-1$ 
distinct elements of $\dc_H \setminus \aset{S, \emptyset}$
such that ${C_1 \preceq C_2 \preceq \cdots \preceq C_{k-1}}$.
Note that, for any $k$-chain, $C_{k-1} \preceq S$.
A $k$-chain defines the $k$-clustering 
${\aset{C_1, C_2 \setminus C_1, C_3 \setminus C_2, \ldots, C_{k-1} \setminus C_{k-2}, S \setminus C_{k-1}}}$.

\begin{thm}\label{thm:desirable-clustering}
	Let $S$ be a set of trajectories. Every \goodClu{} $k$-clustering 
	of $S$ is defined by a $k$-chain.
\end{thm}
\begin{proof}
	Let $\clus = \aset{C_1,C_2,\ldots,C_k}$ be a \goodClu{} $k$-clustering of $S$.
	Suppose that the clusters in $\clus$ are sorted as in Lemma~\ref{lemma:separating-holes}.	

	Thus, for every $1 \leq i \leq k-1$, we have that $S_i = \cup_{j=1}^i C_j \in \aset{S_{\ell}(h_i),S_r(h_i)}$
	for some uncovered hole $h_i$.  Consider $S_k = S$. 
	Hence, $S_i \in \dc_H$ and $S_i \subseteq S_{i+1}$ for every $1 \leq i \leq k-1$. 
	Therefore, $\aset{\cup_{j=1}^i C_j}_{i=1}^{k-1}$ is a $k$-chain that
	defines $\clus$.
\end{proof}

Notice that $\emptyset$ and $S$ are the source and sink of $D_H$, respectively.
For each~$C$ in~$\dc_H$ and each~$j \geq 1$, let $\sdrec(C,j)$ be the sum of the clusters diameters of an 
optimal \goodClu{} $j$-clustering of $S \setminus C$ for \kcsoma{}. 
The following recurrence holds. 
\begin{equation}\label{eq:min-sum-recorrence}
\sdrec(C,j) = \left\{\begin{array}{rl}
	0, & \mbox{if $C = S$}\\
	\diam(S \setminus C), & \mbox{if $j = 1$}\\
	\displaystyle{\min_{C \preceq C'}\aset{\diam(C' \setminus C) + \sdrec(C', j-1)}}, & 
	\mbox{otherwise.}
	\end{array}\right.
\end{equation}

Given a set of $n$ trajectories $S$ and an integer $k$, with $1 \leq k \leq n$, the minimum sum 
of the diameters of an optimal \goodClu{} $k$-clustering of~$S$ is~${\sdrec(\emptyset, k)}$.
A straightforward dynamic programming implementation of recurrence~\eqref{eq:min-sum-recorrence}
to calculate $\sdrec(\emptyset, k)$ consists in filling a~$|\dc_H| \times k$ matrix, whose
rows represent elements of $\dc_H$ and columns represent values of~$j$ in recurrence~\eqref{eq:min-sum-recorrence}.
To fill in each matrix position from row~$S$ takes time $\Oh(1)$, and from column $j = 1$ takes time $\Oh(n \log n)$, 
which is the time to calculate the diameter of a set of trajectories. The remaining matrix positions
are filled in an order such that, when filling in position $(C, j)$, the values of $\sdrec(C', j')$, 
for all $C \preceq C'$ and $j' \leq j$, are already filled. So it takes time $\Oh(|\dc_H|\, n \log n)$ 
to calculate each of these positions. Therefore, it is possible to calculate $\sdrec(\emptyset, k)$
in time $\Oh(k |\dc_H|^2\, n \log n) = \Oh(n^6 \log n)$.
After the whole matrix is filled, we can construct an optimal \goodClu{}
$k$-clustering of~$S$ in time~$\Oh(k |\dc_H|\, n \log n)$ by tracing back recurrence~\eqref{eq:min-sum-recorrence}.

\subsection{Approximation using \goodClu{} $k$-clusterings}

Let $S$ be a set of trajectories and~$\clus^*$ be an optimal $k$-clustering of $S$.
We say that~$\clus'$ is an optimal \goodClu{} $k$-clustering if~$\sd(\clus')$
is minimal among all \goodClu{} $k$-clusterings of~$S$.
As shown in Figure~\ref{fig:bad-separated}, there are cases in which
no optimal $k$-clustering for \kcsoma{} is \goodClu{}. 

Next, we prove a bound on the value of an optimal \goodClu{} $k$-clustering in terms
of the value of an optimal $k$-clustering. This bound allows one to conclude that the 
algorithm presented in Section~\ref{sec:well-separated-sum} is an $\Oh(k)$-approximation 
for the \kcsoma{}.


\begin{thm}\label{thm:k-approximation-sum}
	Let $S$ be a set of $n \geq 3$ trajectories and $\clus^*$ be an optimal $k$-clustering of~$S$. 
        If~$\clus'$ is an optimal \goodClu{} $k$-clustering of~$S$ for the \kcsoma, then
	\[ \sd(\clus') \leq (1 + \lfloor k/2 \rfloor) \sd(\clus^*). \]
\end{thm}
\begin{proof}
	For $k = 1$, clearly $\sd(\clus^*) = \sd(\clus')$.
	By Lemma~\ref{lemma:sum-separation-uncovered}, this equality
	is also true for~$k = 2$. Thus, we may assume that $k \geq 3$.
	
	By Lemma~\ref{lemma:sum-separation-uncovered}, there is at least one hole 
        of~$S$ uncovered by~$\clus^*$. The holes uncovered by $\clus^*$ partition~$S$
	into~$\dlus = \aset{D_1,D_2,\ldots,D_{\ell}}$ with $\ell < k$.
	Notice that~$\dlus$ is \goodClu{}.
	Moreover, each cluster in~$\clus^*$ is contained in some~$D$ in~$\dlus$.

	We denote by~$\restric{\clus^*}{D}$ the subset of clusters~$\aset{C \in \clus^* \st C \subseteq D}$. 
	For each ${1 \leq j \leq \ell}$, let~$H_j$ be the set of holes of~$S$ that separate a pair of 
        clusters in~$\restric{\clus^*}{D_j}$.
	Each hole in~$H_j$ is covered by~$\clus^*$, since it separates clusters contained in the same~$D$ of~$\dlus$.

	For each $D_j$ in~$\dlus$,
	\[ \diam(D_j) \leq \sum_{C \in \restric{\clus^*}{D_j}} \diam(C) + \sum_{h \in H_j} \area(h).\]
	Hence,
	\begin{align*}
		\sd(\dlus) & = \sum_{j = 1}^{\ell} \diam(D_j) 
		\leq \sum_{C \in \clus^*} \diam(C) + \sum_{j = 1}^{\ell} \sum_{h \in H_j} \area(h)
		\leq \sd(\clus^*) + \frac{k}{2}\sd(\clus^*).
	\end{align*}
	The last inequality comes from the fact that, if a hole~$h$ belongs to some~$H_j$, 
	then~$|\restric{\clus^*}{D_j}| \geq 2$, since~$h$ separates two clusters in~$\restric{\clus^*}{D_j}$. 
        So, each hole of~$S$ can appear in at most~$\lfloor k/2 \rfloor$ different sets $H_j$.

	Therefore, $\sd(\clus') \leq \sd(\dlus) \leq (1 + \lfloor k/2 \rfloor)\sd(\clus^*)$.
\end{proof}

Notice that, if the sets $H_j$ in the proof of Theorem~\ref{thm:k-approximation-sum} were pairwise disjoint, 
each hole in~$\bigcup_{j=1}^{\ell}H_j$ would appear only once in the summation, implying that
\[ \sum_{j = 1}^{\ell} \sum_{h \in H_j} \area(h) \leq \sd(\clus^*). \]
Thus, we would have $\sd(\clus') \leq 2\sd(\clus^*)$.

\section{Minimizing the maximum diameter}
\label{sec:max-area}

Let $\clus$ be a $k$-clustering of a set of trajectories. We denote by
\[\md(\clus) = \max_{C \in \clus} \diam(C)\]
the maximum diameter of $\clus$. The second $k$-clustering problem that we
consider is the following.
\begin{prob}[\kcmax]
	Given a finite set of trajectories $S$ and a positive integer~$k$, find 
	a $k$-clustering $\clus$ of $S$ such that $\md(\clus)$ is minimum.
\end{prob}

Let $S$ be a finite set of trajectories. Let $\clus^*$ be an optimal $k$-clustering of $S$,
i.e., $\md(\clus^*)$ is minimal among all $k$-clusterings of $S$. We denote by 
$\optmd(S,k) = \md(\clus^*)$ the value of an optimal $k$-clustering, which is, in this case,
the maximum diameter of $\clus^*$.

Unlike the \kcsoma{}, there are cases in which no optimal $k$-clustering for \kcmax{}
is defined by a \goodSeq{}. Figure~\ref{fig:max-optimal-no-hole} shows an example
for $k = 2$. The $2$-clustering of the example has no uncovered hole. Thus, this clustering cannot
be defined by a \goodSeq{}. Both clusters have diameter 1 and it is not difficult to
check that any cluster with three trajectories has diameter greater than 1. Also, any
other $2$-clustering has maximum diameter greater than 1. Hence, the $2$-clustering
of the example is optimal for the \kcmax{}. Therefore, we
cannot use the same approach used for the \kcsoma{}, for a constant value of $k$. 
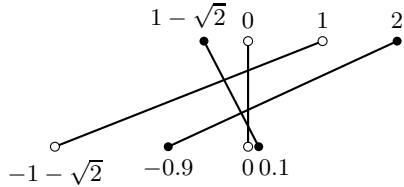
\begin{figure}[ht]
	\center
	\begin{tikzpicture}[scale=1.4]
		\coordinate (l10) at (-1.81421, 0);
		\coordinate (l11) at (0.7, 1.0);
		\coordinate (l20) at (-0.75, 0);
		\coordinate (l21) at (1.4, 1.0);
		\coordinate (l30) at (0, 0);
		\coordinate (l31) at (0, 1.0);
		\coordinate (l40) at (0.1, 0);
		\coordinate (l41) at (-0.41424, 1.0);
		\draw[thick] (l10) -- (l11);
		\draw[thick] (l30) -- (l31);
		\draw[thick] (l20) -- (l21);
		\draw[thick] (l40) -- (l41);
		\segendW{(l10)}{1.2pt};
		\segendW{(l11)}{1.2pt};
		\segend{(l20)}{1.2pt};
		\segend{(l21)}{1.2pt};
		\segendW{(l30)}{1.2pt};
		\segendW{(l31)}{1.2pt};
		\segend{(l40)}{1.2pt};
		\segend{(l41)}{1.2pt};
		\node[below] at ($(l10) - (0, 0.02)$) {\scriptsize $-1-\sqrt{2}$};
		\node[below] at ($(l20) - (0, 0.02)$) {\scriptsize $-0.9$};
		\node[below] at ($(l30) - (0, 0.02)$) {\scriptsize $0$};
		\node[below] at ($(l40) - (-0.15, 0.02)$) {\scriptsize $0.1$};
		\node[above] at ($(l11) + (0, 0.02)$) {\scriptsize $1$};
		\node[above] at ($(l21) + (0, 0.02)$) {\scriptsize $2$};
		\node[above] at ($(l31) + (0, 0.02)$) {\scriptsize $0$};
		\node[above] at ($(l41) + (-0.15, 0.02)$) {\scriptsize $1-\sqrt{2}$};
	\end{tikzpicture}
	\caption{\label{fig:max-optimal-no-hole} Different end types
	represent different clusters. This is the only optimal $2$-clustering for $\kcmax$ and
	it is not defined by any \goodSeq{}.}
\end{figure}

However, if we are specifically looking
for an optimal \goodClu{} $k$-clustering for \kcmax{}, a simple
adaptation of recurrence~\eqref{eq:min-sum-recorrence} can be used,
resulting in an algorithm that finds an optimal \goodClu{}
$k$-clustering for \kcmax{} in time~$\Oh(k |\dc_H|^2\, n \log n)$.
For the \kcmax{} without restriction on the clustering, 
we show some approximation results in the 
next subsections.

\subsection{First approximation}\label{sec:max-first-approximation}

Our first result for the \kcmax{} is
achieved by a reduction to the classical (metric) $k$-center problem.
\begin{prob}[\kcentro]
	Given a positive integer $k$, a finite set $S$, and a distance function $d$ over $S$,
	find a subset $X$ of $S$ such that $|X| \leq k$ and 
	$\max_{s \in S}\min_{x \in X}d(s,x)$ is minimum.
\end{prob}
The \kcentro{} is known to be
NP-hard~\cite{Gonzalez85,KarivHakimi79}, but there are simple 
2-approximations~\cite{Gonzalez85,HochbaumShmoys86}. Moreover, no better
approximation factor is possible unless P~$=$~NP~\cite{Gonzalez85,HsuNemhauser79}.

The elements of the desired set $X$ are called \defi{centers}. We basically want to find
at most $k$ centers such that the largest distance between an element of~$S$ and its nearest
center is minimized. Note that, for each center~$x$ in~$X$, we have an \defi{induced} cluster
\[C_x = \aset{s \in S \st d(s,x) \leq d(s,x')\mbox{, for all } x' \in X}.\]
A $k$-clustering induced by $X$ is roughly $\clus_X = \aset{C_x \st x \in X}$. In fact, at first,~$\clus_X$ 
may not be a $k$-clustering since its clusters are not necessarily disjoint. If this is the case, that is, 
if $C_{x_1} \cap C_{x_2} \neq \emptyset$ for some $x_1$ and $x_2$ in $X$, then we can just remove this 
intersection from either one of the clusters. 

From an instance $(S,k)$ of the \kcmax{}, we build an instance of the \kcentro{} by setting the distance between
two trajectories $s_1$ and $s_2$ in~$S$ as $\diam(\{s_1,s_2\})$. These distances satisfy the triangle inequality,
as shown by Lemma~\ref{lemma:diameter-is-metric}, so the instance is metric and we can apply 
any 2-approximation for the \kcentro{}, outputting a $k$-clustering induced by the~$k$ selected centers.

\begin{lemma}\label{lemma:diameter-is-metric}
	Let $S$ be a set of trajectories. For any three trajectories $a$, $b$, and $c$ of $S$,
	we have $\diam(\aset{a,b}) + \diam(\aset{b,c}) \geq \diam(\aset{a,c})$.
\end{lemma}
\begin{proof}
	Given two trajectories of~$S$, say~$r$ and~$s$, we have that 
	\[
		\diam(\aset{r,s}) = 
		\int_0^1 \! \lpare{\max_{i \in \aset{r,s}} x_i(t) - \min_{i \in \aset{r,s}} x_i(t)} \mathrm{d}t
		= \int_0^1 \! |x_r(t) - x_s(t)|\, \mathrm{d}t.
	\]
	By the linearity of integrals and the subadditivity of modulus, we have that
	\begin{align*}
		\diam(\aset{a,b}) + \diam(\aset{b,c}) & =
		\int_0^1 \! \lpare{|x_a(t) - x_b(t)| + |x_b(t) - x_c(t)|} \mathrm{d}t\\
		& \geq \int_0^1 \! |x_a(t) - x_c(t)|\, \mathrm{d}t\\
		& = \diam(\aset{a,c}).\qedhere
	\end{align*}
\end{proof}

Let $s$ be a trajectory in $S$. To facilitate the reading, in what follows we abuse notation and
consider~$s^t = x_s(t)$ as the position of~$s$ at time~$t$ and also 
as a point ${s^t = (x_s(t),t)}$ in the plane. 
In particular, for~$t = 1/2$ we write~$\bar{s}$ instead of~$s^{1/2}$.
The precise meaning of~$s^t$ will be clear from the context.

\begin{lemma}\label{lemma:max-distance-in-middle}
Let $s$ and $v$ be two trajectories. If $|\bar{s} - \bar{v}| > r$ for some~${r > 0}$, 
then $\diam(\aset{s,v}) > r$.
\end{lemma}
\begin{proof}
	Let $p$ be a trajectory parallel to $s$ that passes through $\bar{v}$. Hence,
	$\diam(\aset{s,p}) = |s^t - p^t| > r$, for any $t \in [0,1]$. If $v = p$,
	we are done. So, suppose that $v \neq p$.

	Since $p$ and $v$ intersect at $t = 1/2$, we have that $|v^0 - p^0| = |v^1 - p^1|$.
	Thus, $\cspan{\aset{v,p}}$ consists of two congruent triangles with vertex
	$\bar{v}$ in common. Assume, without loss of generality, that $\bar{s} < \bar{v}$
	and~$v^0 > p^0$, as shown in Figure~\ref{fig:max-distance-in-middle}.

	\begin{figure}[ht]
	\center
		\begin{tikzpicture}[scale=0.7]
			\coordinate (s0) at (0,0);
			\coordinate (s1) at (1,2);
			\coordinate (v0) at (6,0);
			\coordinate (v1) at (-1,2);
			\coordinate (p0) at ($(s0) + (2,0)$);
			\coordinate (p1) at ($(s1) + (2,0)$);
			\coordinate (vm) at ($(v0)!0.5!(v1)$);
			\fill[gray!35] (p0) -- (vm) -- (v0) -- cycle;
			\fill[gray!35] (v1) -- (intersection of v1--v0 and s1--s0) -- (s1) -- cycle;
			\fill[gray!80] (s1) -- (intersection of v1--v0 and s1--s0) -- (vm) -- (p1) -- cycle;
			\draw[dashed] (vm) -- ($(s0)!0.5!(s1)$);
			\draw[thick] (s0) -- (s1);
			\draw[thick] (v0) -- (v1);
			\draw[thin] (p0) -- (p1);
			\segend{(s0)}{1.7pt}; \node[below] at (s0) {\footnotesize $s^0$};
			\segend{(s1)}{1.7pt}; \node[above] at (s1) {\footnotesize $s^1$};
			\segend{(v0)}{1.7pt}; \node[below] at (v0) {\footnotesize $v^0$};
			\segend{(v1)}{1.7pt}; \node[above] at (v1) {\footnotesize $v^1$};
			\segend{(p0)}{1.7pt}; \node[below] at (p0) {\footnotesize $p^0$};
			\segend{(p1)}{1.7pt}; \node[above] at (p1) {\footnotesize $p^1$};
			\segend{(vm)}{1.7pt}; \node[right] at ($(vm) + (0.03,0.15)$) {\footnotesize $\bar{v}$};
			\node at (3,0.4) {\footnotesize $D$};
			\node at (1.9,1.6) {\footnotesize $B$};
			\node at (0.5,1.82) {\footnotesize $A$};
			\node at (1.4,0.7) {\footnotesize $C$};
		\end{tikzpicture}
	\caption{\label{fig:max-distance-in-middle} Calculating $\diam(\aset{s,v})$ as a function of $\diam(\aset{p,v})$.}
\end{figure}
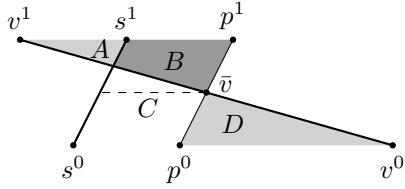
	
	Let $A$, $B$, $C$, and $D$ be the regions indicated in Figure~\ref{fig:max-distance-in-middle}.
	If~$s$ does not intersect~$v$ or if~$s$ intersects~$v$ at one of its extremities, then~$A$
	is empty. Note that 
	\begin{align*}
		\diam(\aset{s,v}) & = \area(A) + \area(C) + \area(D),\\
		\diam(\aset{p,v}) & = \area(A) + \area(B) + \area(D),\\
		\diam(\aset{s,p}) & = \area(B) + \area(C).
	\end{align*}
	Moreover, as the two triangles in the $\cspan{\aset{v,p}}$ are congruent, we have
	$\area(D) = \area(A) + \area(B)$ and, consequently, $\area(D) \geq \area(B)$. 
	Therefore
	\begin{align*}
		\diam(\aset{s,v}) & = \area(A) + (\diam(\aset{s,p}) - \area(B)) + \area(D)\\
		& \geq \diam(\aset{s,p}) > r.\qedhere
	\end{align*}
\end{proof}

\begin{lemma}\label{lemma:max-ball}
	Let $S$ be a set of trajectories and $s$ be a trajectory of $S$. If~$r$ is such that
	$\diam(\aset{s,a}) \leq r$ for all $a$ in $S$, then $\diam(S) \leq (2 + \sqrt{2})r$.
\end{lemma}
\begin{proof}
	Notice that, for any trajectory $a$ such that $|a^0 - s^0| > (1+\sqrt{2})r$, we have
	$\diam(\aset{s,a}) > r$, so $a$ does not belong to $S$. Any trajectory $a$ in $S$
	such that $|a^0 - s^0| > 2r$ intersects $s$ and $|a^1 - s^1| < 2r$. Both observations
	are true if we exchange $a^0 - s^0$ with $a^1 - s^1$.

	By Lemma~\ref{lemma:max-distance-in-middle}, for every trajectory~$a$ in $S$, 
	we have $|\bar{s}-\bar{a}| \leq r$.
	Also, the left side of $S$ is convex and the right side of $S$ is concave with respect to the
	vertical axis. Hence, we can conclude that the \cspan{} of $S$ is inside the polygon~$P$
	highlighted in Figure~\ref{fig:max-ball}.
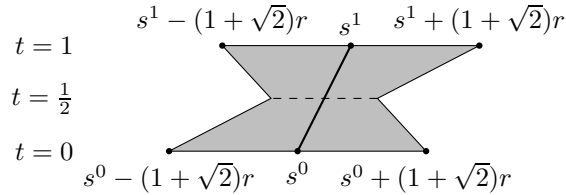
\begin{figure}[ht]
	\center
	\begin{tikzpicture}[scale=0.7]
		\coordinate (c0) at (0,0);
		\coordinate (c1) at (1,2);
		\coordinate (nw) at (-1.41421,2);
		\coordinate (ne) at (3.41421,2);
		\coordinate (sw) at (-2.41421,0);
		\coordinate (se) at (2.41421,0);
		\coordinate (w) at (-0.5,1);
		\coordinate (e) at (1.5,1);
		\draw[fill=gray!50] (nw) -- (ne) -- (e) -- (se) -- (sw) -- (w) -- cycle;
		\draw[dashed] (e) -- (w);
		\draw[thick] (c0) -- (c1);
		\node[above] at (nw) {\footnotesize $s^1 - (1+\sqrt{2})r$};
		\node[above] at (c1) {\footnotesize $s^1$};
		\node[above] at (ne) {\footnotesize $s^1 + (1+\sqrt{2})r$};
		\node[below] at (c0) {\footnotesize $s^0$};
		\node[below] at (se) {\footnotesize $s^0 + (1+\sqrt{2})r$};
		\node[below] at (sw) {\footnotesize $s^0 - (1+\sqrt{2})r$};
		\segend{(c0)}{1.7pt};
		\segend{(c1)}{1.7pt};
		\segend{(nw)}{1.7pt};
		\segend{(ne)}{1.7pt};
		\segend{(sw)}{1.7pt};
		\segend{(se)}{1.7pt};
		\node[left] at (-4,0) {\footnotesize $t = 0$};
		\node[left] at (-4,1) {\footnotesize $t = \frac{1}{2}$};
		\node[left] at (-4,2) {\footnotesize $t = 1$};
	\end{tikzpicture}
	\caption{\label{fig:max-ball}Polygon $P$ containing the \cspan{} of $S$.}
\end{figure}

The area of $P$ is given by the sum of the area of two symmetric trapezoids of height $1/2$:
\[ \area(P) = 2 \left(\frac{1}{2}\cdot\frac{2(1+\sqrt{2})r + 2r}{2} \right) = (2 + \sqrt{2})r. \]
Therefore $\diam(S) \leq \area(P) \leq (2 + \sqrt{2})r$.
\end{proof}

If the maximum distance to a center chosen by a 2-approximation for the \kcentro{} is $r$, 
then $r \leq 2\,\optmd(S,k)$. Indeed, if a cluster
$C$ is such that ${\diam(C) \leq q}$, then any two trajectories in~$C$ are at a distance at
most~$q$. Since, in a cluster produced by the algorithm, any two trajectories are at a 
distance at most~$r$ from the center, then~$q \leq 2r$. 
Using the lemma and the bound on~$r$, we deduce that the described algorithm is a
$2(2+\sqrt{2})$-approximation. 

\subsection{Second approximation}\label{sec:max-second-approximation}

Next we describe a better approximation for \kcmax{}, inspired in the bottleneck method
of Hochbaum and Shmoys~\cite{HochbaumShmoys86}. 
First we describe a greedy algorithm that, given a set~$S$ of~$n$ trajectories and a positive number~$D \geq 0$,
obtains a $k'$-clustering~$\clus_D$ with $\md(\clus_D) \leq 2.71\,D$. If $k' > k$, we will have a certificate
that ${\optmd(S,k) > D}$. If $k' = k$, the algorithm gives a~$2.71$-approximation.

Remember that a trajectory $s$ in $S$ is the
\bleftmost{} if~$x_s(0)$ is minimum in~$S$ and, in case of ties,~$x_s(1)$
is minimum within the tied trajectories.
\begin{tabbing}
	\qquad \=\qquad \=\qquad \=\qquad \kill
	\gp{}($S$, $D$) \+\\
		\textbf{if} $S = \emptyset$ \textbf{then return} $\emptyset$\\
		let $s$ be the \bleftmost{} trajectory in $S$\\
		$C \gets \aset{}$\\
		\textbf{for} each $s' \in S$ \textbf{do} \\
			\> \textbf{if} $\diam(\{s, s'\}) \leq D$ \textbf{then} $C \gets C \cup \aset{s'}$\\
		\textbf{return} $\aset{C} \cup \gp{}(S \setminus C,\, D)$
\end{tabbing}

It is clear that the algorithm \gp{} terminates. Since at each call
at least one trajectory is removed from the set of trajectories, the algorithm
does at most $n - 1$ recursive calls. Each call takes time $\Oh(n)$ excluding the recursive call, 
since the diameter of a subset of two trajectories can be calculated in constant time. 
Finding the \bleftmost{} trajectory in $S$ can also be done in constant time if
we preprocess $S$ in time $\Oh(n \log n)$, sorting the trajectories accordingly.
Thus, \gp{} has time complexity $\Oh(n^2)$.

\begin{lemma}
	\label{lemma:max-semi-ball-area}
	Each cluster built by \gp{}$(S,\, D)$ has diameter at most $(4 + \sqrt{2})D/2$.
\end{lemma}
\begin{proof}
	The proof of this lemma is similar to the proof of Lemma~\ref{lemma:max-ball}, but here we have
	a better bound on the \cspan{} of the cluster. 
    Let $s$ be the \bleftmost{} trajectory in $S$ and $C$ be the cluster returned by 
    \gp{}$(S, D)$ that contains $s$.  Any trajectory $u$ in $C$
    has $u^0 \geq s^0$, since~$s$ is the \bleftmost{} one. Also, we have that
    $\diam(\aset{s, u}) \leq D$, so $s^1-(1+\sqrt{2})D \leq u^1 \leq s^1+2D$ and $u^0 \leq s^0 +
    (1+\sqrt{2})D$.  Thus, the \cspan{} of~$C$ is contained in the polygon~$P$ highlighted in 
    Figure~\ref{fig:semi-ball}, whose area is the
    sum of two trapezoids of height~$1/2$:
    \[ \area(P) \leq \frac{((1{+}\sqrt{2})D + 2D)\frac{1}{2}}{2} + \frac{((2D + (1{+}\sqrt{2})D) + 2D)\frac{1}{2}}{2}
	 = \frac{(4{+}\sqrt{2})D}{2}. \]%
	\begin{figure}[ht]
		\begin{center}
			\begin{tikzpicture}[scale=0.7]
				\coordinate (l0) at (0,0);		
				\coordinate (l1) at (0,2);
				\coordinate (nl1) at (-2.4142,2);
				\coordinate (pl0) at (2.4142,0);
				\coordinate (pl1) at (2,2);
				\coordinate (i1) at (1,1);
				\coordinate (i2) at (-1,1);
				\draw[fill=gray!50] (l0) -- (pl0) -- (i1) -- (pl1) -- (nl1) -- (i2) -- cycle;
				\draw[dashed] (i2) -- (i1);
				\draw[thick] (l0) -- (l1); \segend{(l0)}{1.7pt}; \segend{(l1)}{1.7pt};
				\node[above] at (nl1) {\footnotesize $s^1 - (1+\sqrt{2})D$};
				\node[above] at (l1) {\footnotesize $s^1$};
				\node[above] at (pl1) {\footnotesize $s^1 + 2D$};
				\node[below] at (l0) {\footnotesize $s^0$};
				\node[below] at (pl0) {\footnotesize $s^0 + (1+\sqrt{2})D$};
				\segend{(nl1)}{1.7pt};
				\segend{(pl1)}{1.7pt};
				\segend{(pl0)}{1.7pt};
				\node[left] at (-4,0) {\footnotesize $t = 0$};
				\node[left] at (-4,1) {\footnotesize $t = \frac{1}{2}$};
				\node[left] at (-4,2) {\footnotesize $t = 1$};
			\end{tikzpicture}
		\end{center}
		\caption{\label{fig:semi-ball} Polygon containing the span of $C$.}
	\end{figure}
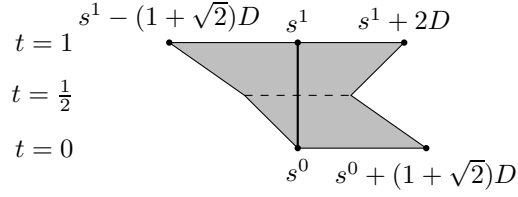%
	Therefore, $\diam(C) \leq \area(P) \leq (4+\sqrt{2})D/2$.
\end{proof}

\begin{lemma}\label{lemma:greed-certificate}
	Let $S$ be a set of trajectories. 
	Let $k$ be a positive integer and~$D$ a positive real.
	If \gp{}$(S,k)$ returns a $k'$-clustering with $k' > k$, then $\optmd(S,k) > D$.
\end{lemma}
\begin{proof}
	Let ${\clus = \aset{C_1,C_2,\ldots,C_{k'}}}$ be the clustering returned
	by \gp{}($S$, $D$).
	For $1 \leq i \leq k'$, let $s_i$ be the \bleftmost{} trajectory in~$C_i$. If $k' > k$, then 
	$\aset{s_1,s_2,\ldots,s_{k'}}$ is a certificate that $\optmd(S,k) > D$.

	Let $s_i$ and $s_j$ be any two distinct trajectories in $\aset{s_1,s_2,\ldots,s_{k'}}$. 
	Since~$s_i$ and $s_j$ belong to different clusters, $\diam(\aset{s_i,s_j}) > D$, otherwise 
	\gp{} would have put both in the same cluster. Therefore, any clustering $\clus'$ of~$S$ 
	with $\md(\clus') \leq D$ must have at least $k' > k$ clusters.
\end{proof}

As shown by Lemmas~\ref{lemma:max-semi-ball-area} and~\ref{lemma:greed-certificate}, 
\gp{} is similar to what Hochbaum and Shmoys~\cite{HochbaumShmoys86} call a
relaxed decision procedure for bottleneck problems. However,
we do not know how to restrict the possible values of $\optmd(S,k)$ to a set of polynomial size.
Thus we do an approximate binary search to get the parameter~$D$ as close as we want to $\optmd(S,k)$.
This is done by the algorithm \bsearch{} below, for any $\eps > 0$. 
\begin{tabbing}
	\qquad \=\qquad \=\qquad \=\qquad \=\qquad \kill 
	\bsearch{}($S$, $k$) \+ \\ 
		$a \gets 0$\\ 
		$b \gets \diam(S)$ \\
		$\delta \gets \frac{2\eps}{4+\sqrt{2}}\min_{u,v \in S,\, u \neq v}\diam(\aset{u,v})$\\ 
		\defi{while} ($b - a > \delta$) \defi{do} \+ \\ 
			$D \gets \frac{a + b}{2}$ \\
			$\clus \gets \gp{}(S, D)$ \\ 
			\defi{if} $|\clus| > k$ \defi{then} $a \gets D$ \defi{else} $b \gets D$ \\ 
		\- \kill 
		\defi{if} $|\clus| > k$ \defi{then return} $\gp{}(S, b)$ \defi{else return} $\clus$
\end{tabbing}

After~$i$ iterations, the search interval is halved $i$ times
and, thus, has length $2^{-i}\diam(S)$. 
So, it takes at most $\log \diam(S) - \log \delta$ iterations, for~$\delta$ 
as in the algorithm, to decrease the search interval length to~$\delta$.

\begin{thm}\label{thm:max-area-approx}
	Let $S$ be a set of $n$ trajectories and $k$ be a positive integer.
	For every $\eps > 0$, \bsearch{}$(S, k)$ is a $((4+\sqrt{2})/2 +
    \eps)$-approximation for \kcmax{}. 
\end{thm}
\begin{proof}
	First, notice that the algorithm always returns a $k$-clustering, 
	since the value of $b$ in the algorithm starts with $\diam(S)$. 
	Also, from the binary search, we know that $\optmd(S,k)$ is in the 
	interval~$[a,b]$ and at the end of the search we have $b - a \leq \delta$. 
	Thus, $b \leq \optmd(S,k) + \delta$.

	Suppose that $k < n$, otherwise it is trivial to find an optimal solution. 
	Hence, one of the clusters in any optimal $k$-clustering must have at least 
	two trajectories and $\optmd(S,k) \geq \min_{u,v \in S,\, u \neq v}\diam(\aset{u,v})$.
	
	Since $\delta = (2\eps/(4+\sqrt{2}))\min_{u,v \in S,\, u \neq v}\diam(\aset{u,v})$,
	\begin{align*}
		\md(\clus) & \leq \frac{4 + \sqrt{2}}{2}\,b \leq \frac{4 + \sqrt{2}}{2}(\optmd(S,k) + \delta)\\
		& = \frac{4 + \sqrt{2}}{2}\optmd(S,k) + 
		\frac{4 + \sqrt{2}}{2} \frac{2\eps}{4+\sqrt{2}}\min_{\substack{u,v \in S\\ u \neq v}}\diam(\aset{u,v})\\
		& = \frac{4 + \sqrt{2}}{2}\optmd(S,k) + \eps\,\min_{\substack{u,v \in S\\ u \neq v}}\diam(\aset{u,v})\\
		& \leq \left(\frac{4 + \sqrt{2}}{2} + \eps\right)\optmd(S,k).
	\end{align*}

	By the choice of $\delta$, \bsearch{}($S$, $k$) terminates within
	at most $\log \diam(S) + \log (4 + \sqrt{2}) - \log (2\eps \min_{u,v \in S,\, u\neq v}\diam(\aset{u,v}))$ 
	iterations. Each iteration consists of a call to \gp{} which takes time~$\Oh(n^2)$. 
	Thus, for a fixed~$\eps$, the algorithm runs in time polynomial in the input size.

	Therefore, \bsearch{}($S$, $k$) is a ${((4+\sqrt{2})/2 + \eps)}$-approximation for \kcmax{},
	for every $\eps > 0$.
\end{proof}

Notice that the value of~$\eps$ is not considered as part of the input when we say that \bsearch{} runs in time 
polynomial in the input size.  Fortunately, in the case of \bsearch{}, the time complexity is not only polynomial 
in the input size but also in~$1/\eps$.

\section{Final comments}

The formalization and study of kinetic versions of clustering problems seems quite intriguing
and challenging.  Of course it would be nice to address variants of the problems we addressed 
in higher dimensions, considering points moving in the plane, or in a 3D space.  Also, allowing 
the clusters to change with time, in a smooth way, seems reasonable and leads to interesting 
questions. 

We focused on the one-dimensional case, presenting polynomial-time algorithms for the \kcsoma{}
under some assumptions, and approximation algorithms for the \kcmax{}.  However, the complexity 
of the {\kcsoma{}} and the \kcmax{} remains open, that is, we do not know whether the \kcsoma{} 
and the \kcmax{} are NP-hard.  It would be nice to settle the complexity of these two problems, 
either by proving that they are NP-hard, or by presenting polynomial-time algorithms to solve 
them.  Meanwhile, achieving better approximations for \kcmax{} and for \kcsoma{} would also be 
nice.  In particular, as far as we know, it is possible that the algorithm described at the end 
of Section~\ref{sec:min-sum-areas}, that outputs an optimal well-separated $k$-clustering, 
achieves a constant approximation ratio for \kcsoma{}, that is, independent of the value of~$k$.

\bibliographystyle{amsalpha}
\bibliography{dyn-clus}

\newcommand{\etalchar}[1]{$^{#1}$}
\providecommand{\bysame}{\leavevmode\hbox to3em{\hrulefill}\thinspace}
\providecommand{\MR}{\relax\ifhmode\unskip\space\fi MR }
\providecommand{\MRhref}[2]{%
  \href{http://www.ams.org/mathscinet-getitem?mr=#1}{#2}
}
\providecommand{\href}[2]{#2}
\begin{thebibliography}{GGH{\etalchar{+}}03}

\bibitem[Ata85]{Atallah85}
M.~Atallah, \emph{{Some dynamic computational geometry problems}}, Comput.
  Math. Appl. \textbf{11} (1985), no.~12, 1171--1181.

\bibitem[BGH99]{BaschGuibasHershberger99}
J.~Basch, L.~Guibas, and J.~Hershberger, \emph{{Data Structures for Mobile
  Data}}, J. Algorithms \textbf{31} (1999), 1--28.

\bibitem[Bru78]{Brucker78}
P.~Brucker, \emph{{On the Complexity of Clustering Problems}}, Optimization and
  Operations Research, LNEMS, vol. 157, Springer Berlin Heidelberg, 1978,
  pp.~45--54 (English).

\bibitem[GGH{\etalchar{+}}03]{Gao+03}
J.~Gao, L.~Guibas, J.~Hershberger, L.~Zhang, and A.~Zhu, \emph{{Discrete Mobile
  Centers}}, Discrete Comput. Geom. \textbf{30} (2003), no.~1, 45--63
  (English).

\bibitem[Gon85]{Gonzalez85}
T.~Gonzalez, \emph{Clustering to minimize the maximum intercluster distance},
  Theor. Comput. Sci. \textbf{38} (1985), 293--306.

\bibitem[HN79]{HsuNemhauser79}
W.L. Hsu and G.~Nemhauser, \emph{{Easy and hard bottleneck location problems}},
  Discrete Appl. Math. \textbf{1} (1979), no.~3, 209--215.

\bibitem[HP04]{Har-Peled04}
S.~Har-Peled, \emph{{Clustering Motion}}, Disc. Comput. Geom. \textbf{31}
  (2004), no.~4, 545--565 (English).

\bibitem[HS86]{HochbaumShmoys86}
D.~Hochbaum and D.~Shmoys, \emph{{A Unified Approach to Approximation
  Algorithms for Bottleneck Problems}}, J. ACM \textbf{33} (1986), no.~3,
  533--550.

\bibitem[KH79]{KarivHakimi79}
O.~Kariv and S.~Hakimi, \emph{An {A}lgorithmic {A}pproach to {N}etwork
  {L}ocation {P}roblems. {I}: {T}he $p$-{C}enters}, SIAM J. Appl. Math.
  \textbf{37} (1979), no.~3, 513--538.

\bibitem[LHW07]{LeeHanWhang07}
J.G. Lee, J.~Han, and K.Y. Whang, \emph{{Trajectory Clustering: A
  Partition-and-group Framework}}, SIGMOD'07, ACM, 2007, pp.~593--604.

\bibitem[O'R98]{ORourke98}
J.~O'Rourke, \emph{{Computational Geometry in C}}, Cambridge Tracts in
  Theoretical Computer Science, Cambridge University Press, 1998.

\bibitem[RD69]{RennieDobson69}
Basil~Cameron Rennie and Annette~Jane Dobson, \emph{{On Stirling Numbers of the
  Second Kind}}, {Journal of Combinatorial Theory} \textbf{7} (1969), no.~2,
  116--121.

\end{thebibliography}

\end{document}